\documentclass[preprint]{elsarticle}

\usepackage{array,xspace,multirow,hhline,tikz,colortbl,tabularx,booktabs,fixltx2e,amsmath,amssymb,amsfonts,amsthm}
\usetikzlibrary{arrows}
\usepackage[ruled,linesnumbered, noend]{algorithm2e}
\usepackage{eqparbox}
\usepackage{verbatim,ifthen}
\usepackage{enumitem}
\usepackage{pifont}
\usepackage{ifthen}
\usepackage{calrsfs,mathrsfs}
\usepackage{bbding,pifont}
\usepackage{pgflibraryshapes}

\usepackage{gensymb} 

	\usepackage{varioref}
		\usepackage{subfigure}

\definecolor{light-gray}{gray}{0.9}

\theoremstyle{theorem}
	\newtheorem{lemma}{Lemma}%
\newtheorem{theorem}{Theorem}%

\theoremstyle{remark}
\newtheorem{remark}{Remark}
\newtheorem{corollary}{Corollary}%
\newtheorem{example}{Example}

	\newcommand\eat[1]{}

	\usepackage{enumitem}
	\setenumerate[1]{label=\rm(\it{\roman{*}}\rm),ref=({\it\roman{*}}),leftmargin=*}
	\newlength{\wordlength}

	\newcommand{\eqclass}[2][]{\ifthenelse{\equal{#1}{}}{[#2]}{[#2]_{\sim_{#1}}}}

	\newcommand{\squared}[1]{\fbox{#1}}

\usepackage{enumitem}
\setenumerate[1]{label=\rm(\it{\roman{*}}\rm),ref=({\it\roman{*}}),leftmargin=*}

\newcommand{\nbh}[1][]{
	\ifthenelse{\equal{#1}{}}{\nu}{\nu(#1)}
}

\newcommand{\cstr}[1][]{
	\ifthenelse{\equal{#1}{}}{\mathscr S}{\cstr(#1)}
}

\newcommand{\choice}[1][]{
	\ifthenelse{\equal{#1}{}}{\mathit{C}}{\choice(#1)}

		\newcommand{\ml}[1][]{\ensuremath{\ifthenelse{\equal{#1}{}}{\mathit{ML}}{\mathit{ML}(#1)}}\xspace}
		\newcommand{\sml}[1][]{\ensuremath{\ifthenelse{\equal{#1}{}}{\mathit{SML}}{\mathit{SML}(#1)}}\xspace}
		\newcommand{\sd}[1][]{\ensuremath{\ifthenelse{\equal{#1}{}}{\mathit{SD}}{\mathit{SD}(#1)}}\xspace}
		\newcommand{\rsd}[1][]{\ensuremath{\ifthenelse{\equal{#1}{}}{\mathit{RSD}}{\mathit{RSD}(#1)}}\xspace}
		\newcommand{\rd}[1][]{\ensuremath{\ifthenelse{\equal{#1}{}}{\mathit{RD}}{\mathit{RD}(#1)}}\xspace}
		\newcommand{\st}[1][]{\ensuremath{\ifthenelse{\equal{#1}{}}{\mathit{ST}}{\mathit{ST}(#1)}}\xspace}
		\newcommand{\bd}[1][]{\ensuremath{\ifthenelse{\equal{#1}{}}{\mathit{BD}}{\mathit{BD}(#1)}}\xspace}
		\newcommand{\pc}[1][]{\ensuremath{\ifthenelse{\equal{#1}{}}{\mathit{PC}}{\mathit{PC}(#1)}}\xspace}
		\newcommand{\dl}[1][]{\ensuremath{\ifthenelse{\equal{#1}{}}{\mathit{DL}}{\mathit{DL}(#1)}}\xspace}
		\newcommand{\ul}[1][]{\ensuremath{\ifthenelse{\equal{#1}{}}{\mathit{UL}}{\mathit{UL}(#1)}}\xspace}

			\newcommand{\indiff}{\ensuremath{\sim}}}

\begin{document}

\title{A polynomial-time algorithm for computing a Pareto optimal and almost proportional allocation}



	\author{Haris Aziz\corref{cor1}} \ead{haris.aziz@unsw.edu.au}
		\address{UNSW Sydney and Data61 CSIRO, Australia}
		
		\author{Herv{\'{e}} Moulin\corref{cor1}} \ead{herve.moulin@glasgow.ac.uk}
		\address{University of Glasgow, Glasgow, UK and Higher School of Economics, St. Petersburg, Russia}

		\author{Fedor Sandomirskiy\corref{cor1}} \ead{sandomirski@yandex.ru}
		\address{Technion, Haifa, Israel and Higher School of Economics, St. Petersburg, Russia}



	\begin{keyword}
	 Fair Division\sep
	Pareto optimality\sep
	Proportionality\sep
		
		\emph{JEL}: C62, C63, and C78
	\end{keyword}

	\begin{abstract}
		We consider fair allocation of indivisible items under additive utilities. We show that there exists a strongly polynomial-time algorithm that always computes an allocation satisfying Pareto optimality and proportionality up to one item even if the utilities are mixed and the agents have asymmetric weights.  The result does not hold if either of Pareto optimality or PROP1 is replaced with slightly stronger concepts. 
	\end{abstract}

	\maketitle


		\sloppy


		\section{Introduction}

	We consider fair allocation of indivisible items under additive utilities. 
	For an agent, an item can be a good (yielding positive utility) or a chore (yielding negative utility). Fair allocation of indivisible items has received renewed interest since it was proved that a \emph{Pareto optimal (PO)} and \emph{envy-free up to one item (EF1)} allocation exists for positive utilities~\citep{CKM+16}. However, the existence and complexity of a Pareto optimal (PO) and EF1 allocation is open when utilities may be negative. The complexity of computing such an allocation is also open for the case of positive utilities. 
	In view of these open questions, a natural relaxation of EF1 called \emph{proportionality up to one item (PROP1)} has started to receive deeper interest. PROP1 requires each agent gets utility that is at least her proportionality guarantee if she loses her biggest chore or alternatively obtain the biggest good allocated to some other agent. 

	Interestingly, even the existence and complexity of PROP1 and Pareto optimal allocation has been an open problem when the utilities are mixed or even negative (see e.g. {a recent survey by} \citet{FrSh19a}). We study this central problem. In previous work, the existence of PROP1 and PO allocations has been established only in the context of goods (positive utilities) and very recently for the case of chores.

	For the case of goods, \citet{CFS17a} posed the complexity of computing a PROP1 and PO allocation as an open problem. They had proved that a PROP1 and PO outcomes always exists even for a public decision making setting that is more general than allocation of indivisible goods. 
	\citet{BaKr19a} presented a {strongly} polynomial-time algorithm that always finds a PROP1 and PO allocation for positive utilities. \citet{BrSa19a} proved that there exists a strongly polynomial-time algorithm for chores that always finds a weighted PROP1 and PO allocation if the number of agents or items is fixed. 
	For mixed utilities, \citet{ACI+18} presented a strongly polynomial-time algorithm to compute a PROP1 and PO allocation when the number of agents is two. They also present a strongly polynomial-time algorithm to compute an EF1 allocation for any number of agents if the preference relation satisfies double monotonicity.

	\paragraph{Contribution}
	We show that even for the case of mixed utilities and any number of agents, an fPO (property stronger than PO) and PROP1 allocation always exists. In particular, we design a strongly polynomial-time algorithm that achieves fPO and PROP1 even if the number of agents or items is not fixed, the utilities are mixed and the agents have asymmetric weights. We obtain as corollaries several recent results that have been proved for only goods or only chores or for weaker requirements. 
	
	\paragraph{Method} {Our results are based on the rounding argument: we first compute a proportional fPO allocation with divisible items and then round it in a clever way that preserves fPO property and ensures PROP1. The known results for goods~\cite{BaKr19a} and for chores~\cite{BrSa19a} use the same methodology; however, they heavily rely on the concept of competitive equilibrium with equal incomes (CEEI) both for computing the initial divisible allocation and for the rounding part. Currently, no algorithms are known for CEEI in economies with mixture of goods and chores, which makes the case of mixed items special and existing approaches inapplicable. We circumvent this difficulty by constructing a rounding procedure that does not rely on equilibrium prices and is applicable to any fPO proportional allocation of divisible items with acyclic consumption graph. To compute such an input allocation we start from the equal division and then find a Pareto-dominating allocation by conducting sequential cyclic trades, the old economic insight recently embodied as an algorithm by~\citet{Sase19a}.
}
	



		\section{Preliminaries}
	
		We consider the allocation of $m$ items in set $O$ to $n$ agents in set $N$. Each agent $i\in N$ has a weight {$b_i>0$} where $\sum_{i\in N}b_i=1$.

	Each item is allocated fully. If each item is allocated to exactly one agent, we call the allocation\emph{ integral}. We will generally denote a fractional allocation by $x=(x_1,\ldots, x_n)$ where $x_i$ is the allocation of agent $i$ and $x_{i,o}$ is the fraction of item $o$ given to agent $i$. We will typically denote an integral allocation by $\pi$ where $\pi_i$ denotes the allocated set of items of agent $i$.

	{By $u_i(o)$ we denote the agent $i$'s utility of receiving the whole item~$o$. The utilities $u_i(o)$ may have mixed signs: an item $o$ can be a chore for some $i$ ($u_i(o)<0$), a good for another agent $j$ ($u_j(o)>0$), and a neutral item for some agent $k$ ($u_k(o)=0$). Agents have additive utilities over allocations: $u_i(x_i)=\sum_{o\in O} u_i(o)x_{i,o}$ and similarly  $u_i(\pi_i)=\sum_{o\in \pi_i} u_i(o)$ in case of integral allocation.}

	{An allocation $y$  Pareto improves an allocation $x$ if $u_i(y_i)\geq u_i(x_i)$ for all $i\in N$ and for some $i$ the inequality is strict.}
	We will call an integral allocation \emph{Pareto optimal (PO)}, if {no integral allocation improves it.}	
	An allocation that {cannot be improved by any fractional allocation} is called \emph{fPO}. {Clearly, an fPO integral allocation is PO as well.}

	An allocation {$x$} is \emph{weighted PROP}  if for each agent $i\in N$, {$u_i(x_i)\geq u_{i}(O)b_i$.}
	An integral allocation $\pi$ is \emph{weighted PROP1} if for each agent $i\in N$, 
	\begin{itemize} 
	\item[$\bullet$] $u_i(\pi_i)\geq u_{i}(O)b_i$; or
	\item[$\bullet$] $u_i(\pi_i)+u_i(o)\geq u_{i}(O)b_i$ for some $o\in O\setminus \pi_i$; or 
	\item[$\bullet$] $u_i(\pi_i)-u_i(o)\geq u_{i}(O)b_i$ for some $o\in \pi_i$. 
	\end{itemize}

	In the literature, PROP1 was studied with respect to goods by \citet{CFS17a}. It has recently been considered for mixed utilities~\citep{ACI+18}.

	When $b_i=1/n$ for all $i\in N$, weighted PROP1 is equivalent to PROP1.

	\medskip

	For any fractional or integral allocation $x$, the corresponding {\emph{consumption graph} $G_x$} is a bipartite graph with vertices $(N\cup O)$ and the edge set $E=\{\{i,o\}\mid x_{i,o}>0\}.$ If an agent $i$ shares an item with an agent $j$ we call $j$, agent $i$'s \emph{neighbor.} 
	
		\section{Algorithm}

		We show that a PO and weighted PROP1 always exists and it can be computed in strongly-polynomial time even for mixed utilities. One possible algorithmic approach for achieving such allocations is to start from a PROP1 allocation and find a Pareto optimal Pareto improving allocation. However, finding an integral allocation that is a Pareto improvement over another integral allocation is generally a computationally hard problem (see e.g. \citet{ABL+19a} and \citet{KBKZ09}) even for the case of goods. Also, a Pareto improvement over a PROP1 allocation may not even satisfy PROP1.
	We provide explicit examples (Examples~\ref{example:goods} and \ref{example:chores}) for this phenomenon.

		\vspace{2em}
	
	In view of the challenges encountered in finding Pareto improvements while maintaining PROP1, we take another route that has been popularized recently (see e.g. \citet{BaKr19a} and \cite{Sase19a}): deal with fractional allocations that are Pareto improving and convert them to suitable integral allocations. 
	
	Our idea is to start with a fractional proportional allocation $x^{\text{prop}}$ and then find a fractional  Pareto optimal allocation $x$ that Pareto improves $x^{\text{prop}}$. 
	We ensure that the  consumption graph of $x$ is acyclic. The acyclicity of the consumption graph is critically used to carefully round $x$ into an integral allocation. The algorithm is described as Algorithm~\ref{algo:poprop1}. The way the rounding is done is illustrated in Example~\ref{example:rounding}. {By acyclicity, the consumption graph of $x$ is a collection of trees. The rounding algorithm} picks an agent $i$ who shares some items with other agents, and rounds all her fractions to her advantage: give to $i$ any good $a$ he shares, and give any chore $b$ to someone with whom he was sharing. In the subtree starting at an agent $j$ who was sharing $a$ or $b$ with $i$, break all the other partial shares of $j$ to her advantage. And so on. The acyclicity of the tree guarantees that this algorithm terminates and returns a {PROP1} allocation. 

			\begin{algorithm}[h!]
				\DontPrintSemicolon
				\KwIn{An instance $I = (N,O,u,b)$}
				\KwOut{Integral allocation {$x^*$}}
				Start with a proportional allocation $x^{\text{prop}}$ that gives a share {$b_i$} of each item to each {agent~$i$}.\;
				Find an fPO fractional allocation $x$ that Pareto dominates $x^{\text{prop}}$ and has {an acyclic  consumption graph $G_{x}$   	(computed via the algorithm from Lemma~2.5 in~(\citet{Sase19a})).}   \label{step:acyclic}\;
				Round the fractional allocation $x$ into an integral allocation $x^*$ as follows:\;
				If some $j$ shares an item $o$ for which $u_j(o)=0$, we give it fully to an agent $i$ who shares $o$ in $x$. {Update $x$ so that $x_{i,o}=1$ and $x_{j,o}=0$ for all $j\neq i$.} \label{step:zero_items}\;
			$Q\leftarrow\emptyset$, an empty FIFO (First-In-First-Out) queue of agents\;
					\While{there is an agent $i$ sharing  {at least one item $o$} with others \label{step:while_begin}}{ 						
					Add $i$ to $Q$ \label{step:addinQ} \;
					\While{$Q$ is non-empty}{
						Take the first agent $j$ out of $Q$\;
						Add all the neighbors of $j$ to the end of $Q$ \label{step:queue_add}\;
						\For{each $o$ shared by $j$}{
						\If{$u_j(o)>0$}{give $o$ fully to $j$}
						\ElseIf{$u_j(o)<0$}{give $o$ to a
							 neighbor with whom $o$ is shared \label{step:else}}
						Update $x$ \label{step:while_end}\;
				    }	 	
				    }
			    }
		 \Return $x^*=x$ \;
		 	 {\tcc{				 For definiteness, we use the following tie-braking conventions.   Zero item $o$ is given to the lowest-index $i$ who shares $o$ on step~\ref{step:zero_items}. 
		 		While-cycle~\ref{step:while_begin} takes agent $i$ who shares \emph{exactly} one item  (such $i$ exists by acyclicity of $G_x$ as long as there is at least one shared item); if there are multiple such $i$,  the lowest-index agent is chosen.  On step~\ref{step:queue_add}, lowest-index agents enter the queue first.  A chore $o$ is given to a lowest-index neighbour with whom it is shared (step~\ref{step:else}).}}

				\caption{Algorithm to find a weighted PROP1 and PO allocation}
				\label{algo:poprop1}
			\end{algorithm}

		\begin{lemma}\label{lm:1}
			Algorithm~\ref{algo:poprop1} returns an integral allocation that is weighted PROP1 in time $O(n^2m^2(n + m))$.
			\end{lemma}
			\begin{proof}
			{Let us show that the returned allocation $x^*$ satisfies weighted PROP1. Since the divisible allocation $x$ from the Step~\ref{step:acyclic} is an fPO allocation that Pareto dominates $x^{\text{prop}}$, it satisfies weighted PROP. This property is preserved when shared zero items are allocated entirely to one of the owners at Step~\ref{step:zero_items}. }
				
				{While-cycles (Steps~\ref{step:while_begin} to~\ref{step:while_end}) implement breadth-first exploration of the consumption graph combined with reallocation of shared {items} (the inner cycle explores agents in each connected component and outer cycle allows to switch between components).
	The algorithm touches no item  consumed fully by one agent; only shared items are reallocated entirely to one of the original partial-owners. Thus to ensure that $x^*$ is weighted PROP1,
				it is enough to show that any agent $j$ can lose at most one partially consumed good (an {item} $o$ such that $u_j(o)> 0$, $x_{j,o}\in(0,1)$) or get an increased share
of  at most one partially-consumed chore ($u_j(o)< 0$, $x_{j,o}\in(0,1)$) and these two cases are mutually exclusive.}
				
				{Let us call an agent $j$ picked by the internal while-cycle \emph{active}. All the items shared by the active agent are allocated by the for-cycle in her favor ($j$ receives all her shared goods and gets rid of all shared chores); later on her allocation does not change. Therefore, the only possibility for $j$'s utility at $x^*$ to be lower compared to $x$ is if she loses her goods or gets more chores \emph{before} becoming active. 
				This means that at an earlier stage of the algorithm there was an active agent $i$ that shared an {item} $o$ with $j$ and either this {item} $o$ was a good both for $i$ and $j$ (in this case, the algorithm allocated it to $i$ taking away one of $j$'s goods) or $o$ was a chore for both $i$ and $j$ and was allocated to $j$ since she was the lowest-index neighbor of $i$. We call such $i$ a \emph{predecessor} of $j$.}
					
				{By acyclicity of the consumption graph, $j$ can have at most one predecessor. Indeed, the if $j$ has two predecessors $i$ and $i'$, then there is a path in the original consumption graph connecting $i$ and $i'$ (inside the outer while-cycle, only neighbors of previously active agents can enter the queue $Q$), but this path cannot pass through $j$ (since $j$ was not active yet); thus the original graph contains a cycle.}
				
				{			The outer while-cycle runs until there are no shared items left; since each iteration of this cycle reallocates at least one shared items, the algorithm terminates in finite number of steps and outputs an integral allocation $x^*$. Since no agent has more than one predecessor, we conclude that  $x^*$ satisfies weighted PROP1. 
				 }


		{	It remains to estimate the time-complexity.	The allocation $x$ is computed in time $O(n^2m^2(n + m))$ according to the algorithm of 
	\citet{Sase19a}. After that, the consumption graph $G_x$ is computed in time {$O(n\cdot m)$}. {It has $n+m$ vertices $V$,  and at most $n+m-1$ edges $E$ since, by acyclicity, $G_x$ is a collection of trees.  ``Root agents'' $i$ from the outer while-cycle, one per each subtree of $G_x$, can be found in $O(|V|+|E|)=O(n+m)$ by the depth-first search. For each such agent $i$, the internal while-cycle represents breadth-first exploration of the tree $T_i=(V_i,E_i)$ rooted at $i$ combined with reallocation of each discovered shared item (takes a constant time). For each tree, breadth-first search takes time $O(|V_i|+|E_i|)$ and hence the whole outer while-cycle can be implemented in $O(|V|+|E|)=O(n+m)$. Thus the overall time-complexity is determined by the Pareto improvement phase of the algorithm.}
}
%
%
%
		\end{proof}
	\begin{lemma}\label{lm:2}
			Algorithm~\ref{algo:poprop1} returns an integral allocation that is fPO.
			\end{lemma}
			\begin{proof}
				Allocation $x$ is fPO in Step \ref{step:acyclic}. 
				Since $x$ is fPO, due to known results by \citet{Vari76a}, it maximizes {weighted welfare $\sum_i \lambda_i  u_i(x_i)$ for some strictly positive weighting $\lambda=(\lambda_i)_{i\in N}$} of the agents {(see also Lemma~2.3 in \cite{Sase19a} for a particular case of mixed items)}. After that we modify $x$ to $x^*$ so that $G_{x^*}$ is a subgraph of $G_x$. {Hence each item $o$ is still consumed by agents with highest weighted utility $\lambda_i \cdot u_i(o)$.} Therefore $x^*$ maximizes welfare for the same weighting of the agents as $x$. Thus $x^*$ is fPO as well. 
				\end{proof}
\begin{remark}[Flexibility in the algorithm]\label{rem_flexibility}
	{	Proofs of weighted PROP1 and fPO properties from Lemma~\ref{lm:1} and~\ref{lm:2} do not rely on the details of the exploration procedure. For example, instead of the breadth-first search one can use the depth-first (i.e., the First-In-Last-Out queue $Q$) or start exploring each connected component of the consumption graph from an arbitrary agent (not necessary the one sharing exactly one item). The only condition that is critical is that in each connected component, the set of explored agents (those that were active at some point) must be connected during the exploration process. Together with acyclicity assumption this allows to ensure that no agent has more than one predecessor and thus use the same argument as in the proof  of Lemma~\ref{lm:1} to deduce PROP1. Also, the choice of the lowest-index agents at Steps~\ref{step:zero_items}, {\ref{step:while_begin}}, \ref{step:queue_add}, and~\ref{step:else} is made for tie-breaking purposes only; instead one can pick such an agent randomly or use any other heuristic.}
	
	{This flexibility leads to a family of algorithms that may output different fPO PROP1 allocations.  An interesting question which we leave open is how to pick the best one among them?}
	\end{remark}
			
	Based on the two lemmas we get the following. 

	\begin{theorem}\label{th:main}
		For mixed utilities, there always exists an integral allocation that satisfies weighted PROP1 and fPO. Furthermore, there exists a strongly polynomial-time algorithm in $n+m$ that returns such an allocation.
		\end{theorem}
	
		We obtain several recent results as corollaries of Theorem~\ref{th:main}. 
	
		\begin{corollary}[\citet{ACI+18}]
			 For two agents and mixed utilities, a Pareto-optimal and PROP1 allocation exists and can be computed in strongly polynomial time.
			\end{corollary}
		
			\begin{corollary}[\citet{BaKr19a}]
				 For positive utilities, {an fPO and PROP1 allocation} can be computed in {strongly} polynomial time.
				\end{corollary}
			
				\begin{corollary}[\citet{BrSa19a}]
					 For negative utilities, a Pareto-optimal and weighted PROP1 {allocation} can be computed in strongly polynomial time if the number of agents or items is fixed.
					\end{corollary}
		
				\begin{example}[Illustration of how our algorithm rounds a fractional allocation into an integral allocation]
			\label{example:rounding}

			Consider a fractional allocation $x$ represented in Table~\ref{table:acyc1}.
			{The allocation has the consumption graph  with two connected components depicted in Figure~\ref{fig:graph}. Since the graph is acyclic, it can be viewed as a pair of trees rooted with agent $2$ and $3$.}
			
 			{The fractional allocation is rounded as in Table~\ref{table:acyc2}. First the algorithm picks agent $2$ and gives him both goods $a$ and $d$ that he consumes.  Agent $1$ keeps his good $b$ and passes chore $c$ to agent $5$ who also consumes good $h$. In the second connected component, the algorithm picks agent $3$ and gives him good $f$ entirely.  Agent $4$ gets nothing. }
 			
 			{Note that using the flexibility discussed in Remark~\ref{rem_flexibility}, we can also end up with another rounding. Let us assume that the outer while-cycle picks  agent $1$ first in one connected component and  agent $4$ in the other one. Then agent $1$ gets both goods $a$ and $b$ (previously, the good $a$ was allocated by the algorithm to agent $2$) and passes chore $c$ to agent $5$, who ends up consuming $c$ and $h$. Agent $2$ is left with good $d$ only. In the second component, agent $4$ gets good $f$ (previously, he got nothing) and agent $3$ receives good $g$ and chore $h$.}
 			
 			{By Theorem~\ref{th:main}, both roundings are fPO and PROP1 as long as the original fractional allocation satisfies fPO and weighted PROP.}

			\begin{table}[h!]
				\begin{center}
					\scalebox{1}{
					\setlength{\tabcolsep}{6pt}
					\begin{tabular}{c|cccccccccccccc}
						& $a$ & $b$&$c$&$d$&$e$&$f$&$g$&$h$ \\
						\hline
						$1$ & \squared{+}&\squared{+}&\squared{--}&+&--&+&+&+\\
						$2$ & \squared{+} &-- &--&\squared{+}&--&+&+&+\\
						$3$ &+ &+&--&--&\squared{--}&\squared{+}&\squared{+}&--\\
						$4$ &  --&{+}&--&--&--&\squared{+}&{+}&--\\
						$5$ &-- &-- &\squared{--}&+&--&--&+&\squared{+}
					\end{tabular}
					}
				\end{center}
				\caption{Table for Example~\ref{example:rounding} indicating the signs of the utilities of agents as well as an allocation {before the rounding algorithm is applied.} A square indicates that the agent {consumes a non-zero amount of  the item.} The allocation has an acyclic consumption graph represented in Figure~\ref{fig:graph}. }
				\label{table:acyc1}
				\end{table}

			\begin{table}[h!]
				\begin{center}
					\scalebox{1}{
					\setlength{\tabcolsep}{6pt}
					\begin{tabular}{c|cccccccccccccc}
						& $a$ & $b$&$c$&$d$&$e$&$f$&$g$&$h$ \\
						\hline
						$1$ & {+}&\squared{+}&{--}&+&--&+&+&+\\
						$2$ & \squared{+} &-- &--&\squared{+}&--&+&+&+\\
						$3$ &+ &+&--&--&\squared{--}&\squared{+}&\squared{+}&--\\
						$4$ &  --&{+}&--&--&--&{+}&{+}&--\\
						$5$ &-- &-- &\squared{--}&+&--&--&+&\squared{+}
					\end{tabular}
					}
				\end{center}	
					\caption{The output of the rounding algorithm for the allocation from Table~\ref{table:acyc1}. {Note that each item is consumed by exactly one agent since the the allocation is integral.}}
						\label{table:acyc2}
					\end{table}
					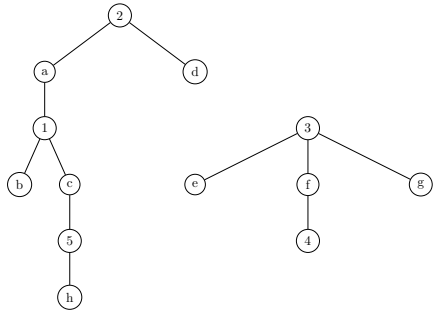
\begin{figure}[h!]
									\scalebox{0.5}{
								\centering
							\begin{tikzpicture}[level/.style={sibling distance=40mm/#1},align=center]
							    \tikzstyle{every node}=[circle,draw]

		

							    \node {2}
							        child {node {a}	
							        child {node {1}  child {node {b}}
							        	child { node {c} child { node {5} child {node {h}}} }}
							        }
						            child{node {d}};

							 \tikzstyle{level}=[sibling distance=30mm]      
							 \node at (5,-3) {3}
							 child{node{e}}
							 child{node{f} child{node{4}}}
							 child{node{g}};       
	
							\end{tikzpicture}
						}
						\centering
						\small
		
						\normalsize
							\caption{{The acyclic consumption} graph corresponding to the allocation in Table~\ref{table:acyc1}.
							}
							\label{fig:graph}

							\end{figure}

				%
				%
				%
				%

						\end{example}
		
			\section{Discussion}
		
	Recently, approaches based on maximin share fairness (a property weaker than proportionality) have been considered for computing fair allocation of indivisible goods {to} \emph{asymmetric} agents~\citep{ACL19a,BNT17a,FHG+19a}. 
	The results in these papers are either for the case of goods or for chores whereas we consider mixed utilities. Our approach uses weighted PROP1 which is a relaxation of the more traditional proportionality guarantee. 

	%
	%
	%
	%
	%


	Our strongly polynomial-time algorithm relies in Step~\ref{step:acyclic} on an algorithm proposed by \citet{Sase19a}. One may wonder whether 
	there is a conceptually {simpler} self-contained algorithm in Step~\ref{step:acyclic}  {achieving an fPO allocation that is an fPO improvement over $x^{\text{prop}}$ and has an} acyclic consumption graph.  Algorithm~\ref{algo:poimprove} satisfies these requirements. It maximizes the sum of utilities subject to proportionality {so the resultant allocation $x^*$} is fPO and proportional. 
	
				\begin{algorithm}[h!]
					\DontPrintSemicolon
					\KwIn{An instance $I = (N,O,u,y)$}
					\KwOut{fPO allocation $x$ that Pareto {improves} allocation $y$ and for which $G_x$ is acyclic}
					$G_x\leftarrow $ complete bipartite graph\;
	$T \leftarrow \emptyset$

	\While{consumption graph $G_x$ has some cycle $C$}{
		Solve $LP_T$ with optimum value $\text{opt}_{T}$:
	\\	$\max\sum_{i\in N} (\sum_{o\in O} u_{i}(o) \cdot x_{i,o}) \text{ s.t. }$
	$$\left\{\begin{array}{lr}
						 \sum_{i\in N}\sum_{o\in O}u_{i}(o) \cdot x_{i,o} \geq \sum_{i\in N}\sum_{o\in O}u_{i}(o) \cdot y_{i,o}& \text{for all $i\in N$}\\
							\sum_{i\in N}x_{i,o} =1 & \text{for all $o\in O$}\\
							 x_{i,o}= 0 & \text{for all $(i,o)\in T$}\\
							 x_{i,o}\geq 0 & \text{for all $i\in N$ and $o\in O$}
							\end{array}
							\right.$$
									\If{there exists some $(i,o)\in C$ such that $\text{opt}_T=\text{opt}_{T\cup \{(i,o)\}}$}{$T\longleftarrow T\cup \{(i,o)\}$}

	}


		 \Return $x=x^*$ \;

					\caption{Algorithm to find a Pareto improvement}
					\label{algo:poimprove}
				\end{algorithm}
	
	{Acyclicity of $G_{x^*}$ is ensured by running the while loop. If an interim fPO allocation $x$ has a cycle $C$ in  $G_x$, there always exists an allocation $x'$ such that all agents get the same utilities and the graph  $G_{x'}$ is a subgraph of $G_x$ but does not contain some  edge $\{i,o\}$  from $C$ (existence of $x'$ follows from a ``cyclic trade'' argument as in the proof of Lemma~2.5 from \citet{Sase19a}). 
	Therefore, we can put the edge $\{i,o\}$ to $T$ without affecting the optimal value of the LP. Thus the consumption graph of the final allocation $x^*$ contains no cycle~$C$.}

	Since Algorithm~\ref{algo:poimprove} uses linear programming, it only gives a guarantee of weakly polynomial-time. Note that instead of Algorithm~\ref{algo:poimprove}, one can {simply solve}
	$LP_{T=\emptyset}$ via the simplex algorithm which returns {a basic feasible solution, i.e, the extreme point of the set of solutions. It is easy to see that for such an extreme solution $x^*$, the graph $G_{x^*}$ is  acyclic. Indeed, if $x^*$  contains a cycle $C$, the allocation can be represented as the convex combination of allocations obtained by ``forward trade'' and ``backward trade'' along $C$.} 
	The simplex algorithm works very well in practice but can in  theory take exponential time in the worst case.

	\paragraph{Extending the result}

	We now point out that our result is un-improvable or challenging to improve in several respects. If PROP1 is replaced by the stronger property of EF1, then it is an open problem whether an EF1 and PO allocation exists or not~\citep{ACI+18}. The problem remains open even for the case of chores.

	If PROP1 is strengthened to a concept called proportionality up to the extreme item (PROPX), then the existence of an allocation satisfying the property is not guaranteed for the case of goods~\citep{Moul18}. We provide a self-contained and simpler example (Example~\ref{example:propx}).

	Finally, one may wonder whether our main result can be strengthened by considering maximum welfare rather than Pareto optimality. However, computing an allocation that is utilitarian-maximal within the set of PROP1 allocations is NP-hard~\citep{AHMS19a}.

\section*{Acknowledgements}

Haris gratefully acknowledges the UNSW Scientia Fellowship and Defence Science and Technology (DST). He also thanks the RIMS-Research Institute for Mathematical Sciences Kyoto for hosting him. 
Herv{\'{e} and Fedor gratefully acknowledge support from the Basic Research Program of the National Research University
Higher School of Economics.
Fedor's work is also supported by the European Research Council (ERC) under the European Union's Horizon 2020 research and innovation program (grant agreement n$\degree$740435) and by Grant 19-01-00762 of the Russian Foundation for Basic Research.

		\bibliographystyle{plainnat}
	 %
	 %

	%
	%

		\appendix
	
		\section{Examples}
	
			\begin{example}[Pareto improvement over a PROP1 allocation may not even satisfy PROP1 when there are goods. ]\label{example:goods}
				Consider the following instance with 3 agents and 31 items. Items in sets $B$ and $C$ are divided respectively into 10 and 20 smaller items each 
			\begin{center}
				\scalebox{1}{
				\setlength{\tabcolsep}{6pt}
				\begin{tabular}{c|cccccccccccccc}
					& $A$ & $B=\{b_1,\ldots, b_{10}\}$&$C=\{c_1,\ldots, c_{20}\}$ \\
					\hline
					$1$ &0.3& 0.2&0.5 \\
					$2$ &0.34& 0.16& 0.5 \\
					$3$ &0.16&0.5& 0.34 \\
				\end{tabular}
				}
			\end{center} 
	
			Initially the allocation $x$ is as follows (indicated via the squares).  Agent 2 and 3 get utility 0.34 which exceeds the proportionality value. Agent 1 gets total utility 0.2 but the utility increases to 0.5 if agent 1 additionally gets item $A$. Therefore, the allocation is PROP1.
	
			\begin{center}
				\scalebox{1}{
				\setlength{\tabcolsep}{6pt}
				\begin{tabular}{c|cccccccccccccc}
					& $A$ & $B=\{b_1,\ldots, b_{10}\}$&$C=\{c_1,\ldots, c_{20}\}$ \\
					\hline
					$1$ &0.3& \squared{0.2}&0.5 \\
					$2$ &\squared{0.34}& 0.16& 0.5 \\
					$3$ &0.16&0.5& \squared{0.34} \\
				\end{tabular}
				}
			\end{center}
	
			Suppose we obtain the following Pareto improving allocation $y$. 
	
			\begin{center}
				\scalebox{1}{
				\setlength{\tabcolsep}{6pt}
				\begin{tabular}{c|cccccccccccccc}
					& $A$ & $B=\{b_1,\ldots, b_{10}\}$&$C=\{c_1,\ldots, c_{20}\}$ \\
					\hline
					$1$ &\squared{0.3}& 0.2&0.5 \\
					$2$ &0.34& 0.16& \squared{0.5} \\
					$3$ &0.16&\squared{0.5}& 0.34 \\
				\end{tabular}
				}
			\end{center}
	
			 Agent 2 and 3 get utility 0.5 which exceeds the proportionality value. Agent 1 gets total utility 0.3 but even if agent 1 is given any other item, the total utility does not exceed $1/3$. Therefore, allocation $y$ is not PROP1. $\diamond$  
		 
		\end{example}

			\begin{example}[Pareto improvement over a PROP1 allocation may not even satisfy PROP1 when there are chores.]\label{example:chores}
				Consider the following instance with 3 agents and 12 items. Items in sets $A$ are divided respectively into 10 smaller items.
			
				\begin{center}
					\scalebox{1}{
					\setlength{\tabcolsep}{6pt}
					\begin{tabular}{c|cccccccccccccc}
						& $A=\{a_1,\ldots, a_{10}\}$ & $B$&$C$ \\
						\hline
						$1$ &--0.4&--0.5&--0.1 \\
						$2$ &--0.3&--0.6&--0.1 \\
						$3$ &--0.6&--0.1&--0.3 \\
					\end{tabular}
					}
				\end{center}

			Initially the allocation $x$ is as follows (indicated via the squares). 
			The allocation is PROP1.  
		
			\begin{center}
				\scalebox{1}{
				\setlength{\tabcolsep}{6pt}
				\begin{tabular}{c|cccccccccccccc}
					& $A=\{a_1,\ldots, a_{10}\}$ & $B$&$C$ \\
					\hline
					$1$ &--0.4&\squared{--0.5}&--0.1 \\
					$2$ &\squared{--0.3}&--0.6&--0.1 \\
					$3$ &--0.6&--0.1&\squared{--0.3} \\
				\end{tabular}
				}
			\end{center}
	
			Suppose we obtain the following Pareto improving allocation $y$. 
	
		\begin{center}
			\scalebox{1}{
			\setlength{\tabcolsep}{6pt}
			\begin{tabular}{c|cccccccccccccc}
				& $A=\{a_1,\ldots, a_{10}\}$ & $B$&$C$ \\
				\hline
				$1$ &\squared{--0.4}&--0.5&--0.1 \\
				$2$ &--0.3&--0.6&\squared{--0.1} \\
				$3$ &--0.6&\squared{--0.1}&--0.3 \\
			\end{tabular}
			}
		\end{center}
	
		Allocation $y$ is not PROP1 because even if agent 1 gets rid of one of the small $A$ items, her utility is $-0.36<-1/3$.
			 $\diamond$  
		\end{example}

		{An integral allocation $\pi$} satisfies \emph{proportionality up to extreme item (PROPX)} if for each agent $i\in N$, 
		\begin{itemize}
		\item[$\bullet$] {$\forall o\in \pi_i$} s.t. $u_i(o)<0$: $u_i(\pi_i\setminus \{o\})\geq u_i(O)/n$; and
		\item[$\bullet$] {$\forall o\notin \pi_i$} s.t. $u_i(o)>0$: $u_i(\pi_i\cup \{o\})\geq u_i(O)/n$.
	\end{itemize}

		\begin{example}[a PROPX allocation may not exist for the case of goods]
				\begin{center}
					\scalebox{1}{
					\setlength{\tabcolsep}{6pt}
					\begin{tabular}{c|cccccccccccccc}
						& $a$ & $b$&$c$&$d$&$e$ \\
						\hline
						$1$ & 3&3&3&3&1\\
						$2$ & 3&3&3&3&1\\
						$3$ & 3&3&3&3&1
					\end{tabular}
					}
				\end{center}
		
				In any most balanced allocation one agent gets two big items, one agent gets one big item (utility $3$) and the small item $e$, and one agent gets only one big item.

				\begin{center}
					\scalebox{1}{
					\setlength{\tabcolsep}{6pt}
					\begin{tabular}{c|cccccccccccccc}
						& $a$ & $b$&$c$&$d$&$e$ \\
						\hline
						$1$ & \squared{3}&\squared{3}&3&3&1\\
						$2$ & 3&3&3&\squared{3}&\squared{1}\\
						$3$ & 3&3&\squared{3}&3&1
					\end{tabular}
					}
				\end{center}
			
				The last agent does not achieve the proportionality value of $13/3>4$ even if she gets the small item. Therefore PROPX is not satisfied. 
				\label{example:propx}
				\end{example}

	\end{document}